\documentclass{stacs_proc}

\usepackage{graphicx}
\usepackage{epsfig}
\usepackage{psfrag}
\usepackage{wrapfig}

\newcommand{\bigO}{O}
\newcommand{\setR}{\mathbb{R}}
\newcommand{\D}{\mathcal{D}}

\newcommand{\C}{\mathcal{C}}
\newcommand{\T}{\mathcal{T}}

\DeclareMathOperator{\st}{s.t.}
\DeclareMathOperator{\intr}{int}
\DeclareMathOperator{\DT}{DT}
\DeclareMathOperator{\opt}{OPT}
\DeclareMathOperator{\con}{conv}

\begin{document}

\title[short title]{title}
\title[Geometric Set Cover and Hitting Sets for Polytopes in $\setR^3$]{Geometric Set
  Cover and Hitting Sets \\ for Polytopes in $\setR^3$}

\author{S\"{o}ren Laue}{S\"{o}ren Laue}
\address{Max-Planck-Institut f\"{u}r Informatik, Campus E1 4, 66123
Saarbr\"{u}cken, Germany}  
\email{soeren@mpi-inf.mpg.de}  
\thanks{This work was supported by the Max Planck Center for Visual
  Computing and Communication (MPC-VCC) funded by the German Federal
  Ministry of Education and Research (FKZ 01IMC01).} 



\begin{abstract}
  \noindent 
  Suppose we are given a finite set of points $P$ in $\setR^3$ and a
  collection of polytopes $\T$ that are all translates of the same
  polytope $T$. We consider two problems in this paper. The first is
  the set cover problem where we want to select a minimal number of
  polytopes from the collection $\T$ such that their union covers all
  input points $P$. The second problem that we consider is finding a
  hitting set for the set of polytopes $\T$, that is, we want to
  select a minimal number of points from the input points $P$ such
  that every given polytope is hit by at least one point.
  
  We give the first constant-factor approximation algorithms for both
  problems. 
  We achieve this by providing an epsilon-net for translates of a
  polytope in $\setR^3$ of size $\bigO(\frac{1}{\epsilon})$. 
\end{abstract}

\maketitle

\stacsheading{2008}{479-490}{Bordeaux}
\firstpageno{479}

\vspace{-0.5cm}
\section*{Introduction}\label{S:one}
Suppose we are given a set of $n$ points $P$ in $\setR^3$ and a collection
of polytopes $\T$ that are all translates of the same polytope $T$. We
consider two problems in this paper. The first is the set cover
problem where we want to select a minimal number of polytopes from the
collection $\T$ such that their union covers all input points $P$. The
second problem that we consider is finding a hitting set for the set
of polytopes $\T$, that is, we want to select a minimal number of
points from the input points $P$ such that every given polytope is hit
by at least one point.   

Both problems, the set cover problem and the hitting set problem which
are in fact dual to each other are very fundamental problems and have
been studied intensively. In a more general setting, where the sets
could be arbitrary subsets, both problems are known to be NP-hard, in
fact they are even hard to approximate within $o(\log
n)$~\cite{LY94}. Even when the sets are induced by geometric objects
it is widely believed that the corresponding set cover problem as well
as the hitting set problem are NP-hard. Several geometric versions of
these problems were even proven to be hard to approximate. Hence, we
are looking for algorithms that approximate both problems. We give the
first constant-factor approximation algorithms for the set cover
problem and the hitting set problem for translates of a polytope in
$\setR^3$. The central idea to our approximation algorithms are small
\emph{epsilon-nets}. 

A set of elements $P$ (also called points) along with a collection
$\T$ of subsets of $P$ (also called ranges) is in general called a
\emph{set system} $(P, \T)$ and for geometric settings also known as
\emph{range spaces}. One essential characteristic of these set systems
is the \emph{Vapnik-Chervonenkis dimension}, or
\emph{VC-dimension}~\cite{VC71}. The VC-dimension is the cardinality
of the largest subset $A\subseteq P$ for which $\{T\cap A : T\in \T\}$ is the powerset
of $A$. If the set $A$ is finite, we say that the set system $(P, \T)$
has bounded 
VC-dimension, otherwise we say the VC-dimension of $(P, \T)$ is 
unbounded. For instance, the set system induced by translates of a
polytope has VC-dimension three as well as the set system induced by
halfspaces in $\setR^2$.  A set $N\subseteq P$ is called an \emph{epsilon-net} for
a given set system $(P, \T)$ if $N\cap T\neq \emptyset$ for every subset $T\in\T$ for
which $\|T\|\geq \epsilon \cdot \|P\|$. In other words, an epsilon-net is a hitting set
for all subsets $T\in \T$ whose cardinality is an $\epsilon$-fraction of the
cardinality of the input point set $P$.  

It is known that there exist epsilon-nets of size
$\bigO\left(\frac{d}{\epsilon} \log \frac{d}{\epsilon}\right)$ for any set system of
VC-dimension $d$~\cite{BEHW89, KPW92}. This bound is in fact tight for
arbitrary set systems as there exist set systems that do not admit
epsilon-nets of size less than this bound~\cite{PW90}. Such an
epsilon-net can be simply found by random sampling~\cite{M}. 

However, for special set systems that are induced by geometric objects
there do exist epsilon-nets of smaller size, namely of size
$\bigO(\frac{1}{\epsilon})$. It has been shown by Pach and
Woeginger~\cite{PW90} that halfspaces in $\setR^2$ and translates of
polytopes in $\setR^2$ admit epsilon-net of size
$\bigO(\frac{1}{\epsilon})$. Matou\v{s}ek et al.~\cite{MSW90} gave an
algorithm for computing  small epsilon-nets for pseudo-disks in $\setR^2$
and halfspaces in $\setR^3$. The result for halfspaces in $\setR^3$ also
follows from a more general statement by Matou\v{s}ek~\cite{M92}.  

Among other reasons for finding epsilon-nets of small size is the fact
that an epsilon-net of size $g(\epsilon)$ immediately implies an
approximation algorithm for the corresponding hitting set with
approximation guarantee of $\bigO(g(1/c)/c)$, where $c$ denotes the
optimal solution to the hitting set~\cite{PA95}. This means, that for
arbitrary set systems of fixed VC-dimension we have an algorithm for
the hitting set problem with approximation $\bigO(\log c)$.  
And for set systems that admit epsilon-nets of size $\bigO(1/ \epsilon)$ we
get an approximation algorithm to the hitting set problem with
constant approximation guarantee.


Clarkson and Varadarajan~\cite{CV05} developed a technique that
connects the complexity of a union of geometric objects to the size of
the epsilon-net for the dual set system. Using this result, they are
able to develop, among other approximation algorithms for geometric
objects in $\setR^2$, a constant-factor approximation algorithm for the
set cover problem induced by translates of unit cubes in $\setR^3$.  

We extend their result to not only the set cover problem but also the
hitting set problem for arbitrary translates of a polytope in
$\setR^3$. We do not require the polytope to be convex or fat. This is the
first constant-factor approximation algorithm for these two
problems. We achieve this by giving an epsilon-net for translates of a
polytope in $\setR^3$ of size $\bigO(\frac{1}{\epsilon})$. We reduce the problem
of finding epsilon-nets for translates of a polytope to a family of
non-piercing objects in $\setR^2$ and then generalize the epsilon-net
finder for pseudo-disks of Matou\v{s}ek et al.~\cite{MSW90} to our
setting.     

The set cover problem
which is studied by Hochbaum and Maass~\cite{HM85} where one is
allowed to move the objects is fundamentally different. They give a
PTAS for their problem.


\section{Small Epsilon-Nets for Polytopes in $\setR^3$}


Let $P$ be a set of $n$ points in $\setR^3$ and let $\T$ be a family of
polytopes  that are all translates of the same bounded polytope
$T_0$.
We want to find a set of polytopes of minimal cardinality among the
collection $\T$ that covers all input points $P$. 
First, we find a small epsilon-net for this set system and use this
later for the constant-factor approximation of the hitting set
problem. Finally, we show how this then can be translated into a
solution for the set cover problem.  

 Throughout this paper we denote by $T$ the polytope as well as
the subset of points from $P$ that $T$ covers and by $\T$ the family
of polytopes as well as the corresponding family of subsets of
$P$. This will make the paper easier to read and it 
will be clear from the context whether we talk about the geometric
object or the corresponding set of points.

\subsection{From Polytopes in $\setR^3$ to Non-Piercing Objects in $\setR^2$}


So given such a set system $(P, \T)$ we want to find an epsilon-net
for it, i.e. we are looking for a set $N\subseteq P$ such that every subset of
points $T\in \T$ with $\|T\|\geq \epsilon \cdot \|P\|$ is stabbed by at least one point
from $N$. 

We can cut the polytope $T$ into, lets say $k$ polytopes $T_1, T_2, 
\ldots, T_k$. If the polytope $T$ contains $\epsilon n$ input points then one
of the polytopes $T_1, T_2, \ldots, T_k$ must contain at least $\frac{\epsilon}{k} \cdot n$
input points. Hence, in order to find an $\epsilon$-net for the set system
$(P, \T)$ 
induced by translates of $T$, it suffices to find $\frac{\epsilon}{k}$-net
for the set systems induced by the translates of $T_1, T_2, \ldots, T_k$.

Following this reasoning we can reduce our problem for finding an
epsilon-net for the set system induced by translates of arbitrary
polytopes to translates of \emph{convex} polytopes by cutting the
possibly non-convex polytope into a set of convex polytopes.
Note that the number of these convex polytopes only depends on the
polytope $T$ and hence is constant for fixed $T$.

Wlog. let $T$ be from now on a convex polytope.
We can place a cubical grid onto the space $\setR^3$ such that for
any translate of $T$ every cubical grid cell contains at most
vertex of $T$. This can be achieved by making the grid fine
enough. Clearly, the maximal number $t$ of grid cells that can be
intersected by $T$ is bounded and only depends on $T$. Again, if
$T$ contains $\epsilon n$ input points then at least one of the cells
must contain at least $\frac{\epsilon}{t}\cdot n$ of the input points. Hence, we
can restrict ourselves to finding epsilon-nets for translates of 
triangular cones
where all input points lie in a cube in $\setR^3$. This just adds a
multiplicative constant to the size of the final epsilon-net.

The case when the cubical cell only contains a halfspace or the
intersection of two halfspaces can be either seen as a special case of
a cone or, in fact, be even treated separately in a much simpler
way. The case of a translate of a halfspace reduces to a
one-dimensional problem an admits an epsilon-net of size 1 and the
case of two intersecting halfspaces reduces to a problem on intervals
which admits an epsilon-net of size $\bigO(1/ \epsilon)$. 

In the following we will construct an epsilon-net for the set system
$(P, \C)$ that is induced by translates of a 
triangular cone $C$. 

Given a cone $C$, we call a set of points $P$ in non-$C$-degenerate
position if every translate of $C$ has at most three points of $P$ on
its boundary. We can always perturb the input points $P$ in such a way
that they are in non-$C$-degenerate position and the collection of
subsets of the form $P\cap C_T$ where $C_T$ is a translate of $C$ does
not decrease~\cite{EW85}. Hence, we can restrict ourselves on
non-$C$-degenerate set of points $P$.  

\begin{wrapfigure}[11]{R}{5cm}
  \begin{center}
    \psfrag{z}{$z=0$}
    \psfrag{C}{$C$}
    \psfrag{r}{$r$}
    \epsfig{file=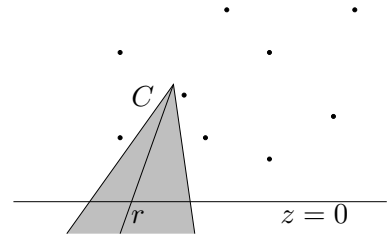,width=5cm}
    \caption{The cone $C$ and its internal ray $r$.}
    \label{fig:cone1}
  \end{center}
\end{wrapfigure}
We place a
coordinate system  such that the input points all have $z$-coordinate
greater than $0$ and a ray $r$ emitting from the apex of the cone $C$ and
lying entirely in the cone should intersect the plane $z=0$. We refer
to such a cone as a cone that \emph{opens to the bottom} and the ray $r$ as
its \emph{internal ray}.
Figure~\ref{fig:cone1} illustrates this setup for the two-dimensional
case.

The following two definitions are helpful generalizations the lower
envelope. 

\begin{defi}
  Given a finite point set $P$ and a triangular cone $C$ that opens to
  the bottom consider the arrangement 
  of all translates of $C$ that have a point of $P$ on its boundary
  but no point of $P$ in its interior. The upper set of plane segments
  that can be seen from above is called the \emph{lower envelope of
    $P$ with respect to cone $C$}.
\end{defi}

Figure~\ref{fig:cone2} illustrates the definition of the lower envelope
in the two-dimensional case.
This definition is similar to the definition of alpha-shapes where
the cone is replaced by a ball.
We call all points that lie on the lower envelope with respect to cone
$C$ \emph{lower envelope points} and denote this set by $L$. 
\begin{figure}[h]
\begin{minipage}{0.49\textwidth}
  \begin{center}
    \psfrag{C}{$C$}
    \epsfig{file=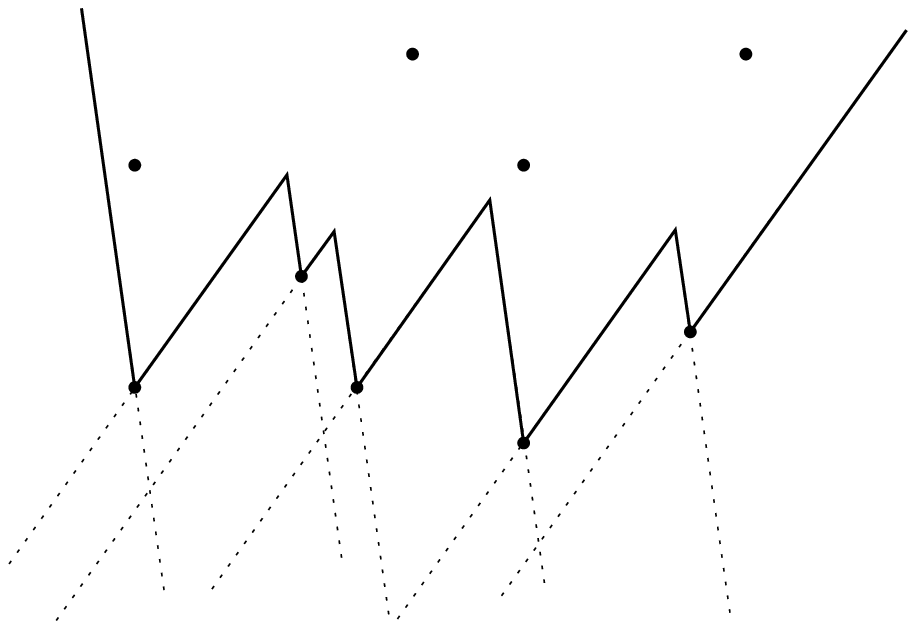,width=6cm}
    \vspace{-3ex}
    \caption{The lower envelope with respect to cone $C$, the
      corresponding cones are drawn dotted.} 
    \label{fig:cone2}
  \end{center}
\end{minipage}%
\hfill
\begin{minipage}{0.49\textwidth}
  \begin{center}
    \psfrag{C}{$C$}
    \epsfig{file=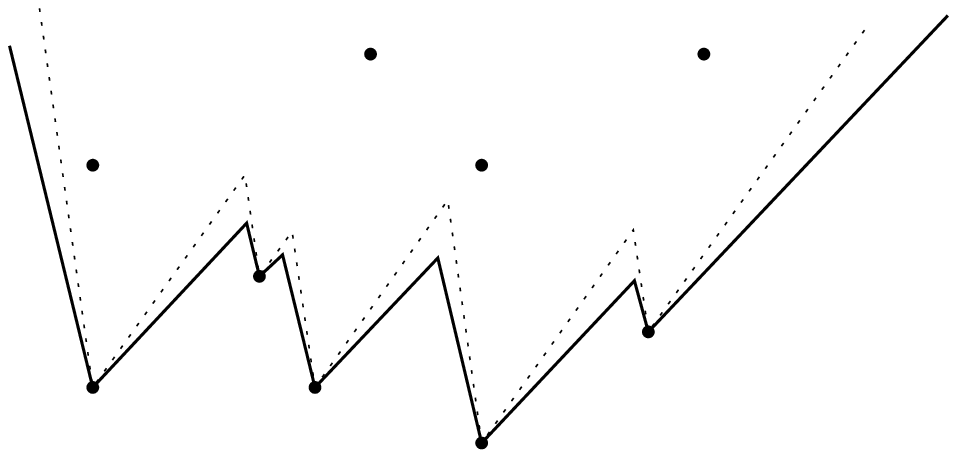,width=6cm}
     \vspace{-3ex}
    \caption{The flattened lower envelope with respect to cone $C$,
      lower envelope is drawn dotted.} 
    \label{fig:cone3}
  \end{center}
\end{minipage}
\end{figure}

\begin{defi}
  Let $C$ be a triangular cone that opens to the bottom and let $P\subseteq
  \setR^3$ be a finite set of points in non-$\C$-degenerate position. Let
  $C'$ be a cone that is flatter that $C$ by small $\delta$ and such that
  it contains $C$ and the combinatorial structure of $P$ and $C'$ is
  the same as for $P$ and $C$. See figure~\ref{fig:cone3} for an
  illustration. Then, the lower envelope of $P$ with respect to $C'$
  is called the \emph{flattened lower envelope of $P$ with respect to cone $C$}.
\end{defi}
Such a cone $C'$ always exists for a finite point set that is in
non-$\C$-degenerate position.
From now on we will abbreviate the term lower envelope with respect to
cone $C$ by lower envelope since we will throughout this paper only
talk about the same cone $C$. The flattened lower envelope can  be
basically seen as a slightly flattened version of the lower envelope.

The next lemma shows that we can reduce the problem of finding an
epsilon-net with respect to cones of arbitrary point sets to lower
envelope points.

\begin{lemma}
  \label{lem:1}
  If for every finite point set $P'\subseteq\setR^3$ of lower envelope points in
  non-$C$-degenerate position there exists 
  an epsilon-net with respect to translates of a cone $C$ of size
  $s(\epsilon)$ then there exists 
  an epsilon-net with respect to translates of a cone $C$ of size
  $3s(\epsilon)$ for every finite point 
  set  $P\subseteq\setR^3$ in non-$C$-degenerate position.
\end{lemma}
\begin{wrapfigure}[12]{r}{5.2cm}
  \begin{center}
    \epsfig{file=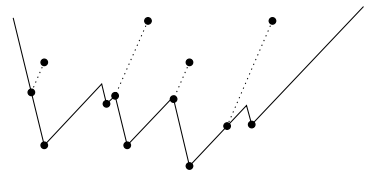,width=5cm}
    \caption{The projection of points onto flattened lower envelope.}
    \label{fig:cone4}
  \end{center}
\end{wrapfigure}
.\vspace{-4ex}
\begin{proof}
  Let $P\subseteq\setR^3$ be such a finite point set in non-$\C$-degenerate
  position and let $C$ denote the cone. Let $L$ denote the set of
  lower envelope points. Let $\bar{L}=P\setminus L$ be the set of all
  non-lower envelope points.   
  We project all non-lower envelope points $\bar{L}$ along the
  internal ray $r$ of cone $C$ onto the flattened lower envelope
  (cf. figure~\ref{fig:cone4}). We denote the projection of a point
  $p$ by $p'$. 
  Let $P'$ be union of the projected points and $L$. Clearly, $P'$ is
  a set of lower envelope points in non-$C$-degenerate position.  

  Suppose we have an epsilon-net $N'$ for this point
  set $P'$. From this epsilon-net $N'$ we will construct an
  epsilon-net $N$ for the original point set $P$. If a point from the set
  $L$ is in the epsilon-net $N'$, we also add it to the epsilon-net
  $N$ for $P$. If however, a projected point $p'$ is in $N'$ then we
  add to $N$ the three points $p_1, p_2$ and $p_3$ from the lower envelope $L$ that determine the cone $C$ on whose boundary also $p'$ lies. Note that whenever an arbitrary cone contains the point $p'$ then it has to contain one of the three points $p_1, p_2$ or $p_3$. 

  We have the following two properties:
  \begin{enumerate}
  \item If a cone contains at least $\epsilon n$ points from the set $P$ then
    it contains at least $\epsilon n$ points from the set $P'$. 
  \item If a cone contains a point from the epsilon-net $N'$ for $P'$
    then the cone contains a point from the epsilon-net $N$ for $P$.  
  \end{enumerate}    
  Both properties prove that the 
  set $N$ is indeed an epsilon-net for $P$.
\end{proof}

The preceding lemma assures that we can restrict ourselves on a finite
set of lower envelope points in non-$C$-degenerate position. For such
a set system we will now construct a corresponding set system of
points in the plane and a collection of regions in the plane.

\begin{defi}
  Let $C$ be a cone and let $P'$ be a finite set of lower envelope
  points in non-$C$-degenerate position and let $\C$ be a collection
  of translates of $C$. We define a projection $\tau$ from the flattened
  lower envelope onto the plane $z=0$ by projecting each point along
  the internal ray $r$. Let the projection of all points $p'\in P'$
  which all lie on the be denoted as the set $S$. For each cone of the
  collection the image of the intersection of the cone with the
  flattened lower envelope is an object $D\subseteq \setR^2$ and the family $\C$
  of cones induces a family of objects which we will denote by $\D$. 
\end{defi}
Using the flattened lower envelope instead of the lower
envelope avoids degeneracy. The intersection of an arbitrary cone with
the flattened lower envelope is always a collection of line
segments. Furthermore, it makes everything continuous in the
sense that if a cone is moved continuously in $\setR^3$ then the
intersection of the cone with the flattened lower envelope moves
continuously as well as its image of the projection $\tau$. Note, that
$\tau$ is injective.  

Analogously, we call a set of points $S\subseteq \setR^2$ in non-$\D$-degenerate
position if every $D\in \D$ has at most three points on its boundary. 
We have the following lemma:
\begin{lemma}
  \label{lem:2}
  If for every finite point set $S\subseteq\setR^2$ in non-$\D$-degenerate
  position there exists 
  an epsilon-net with respect to the family of objects $\D$
  produced by the projection $\tau$ of size $s(\epsilon)$ then there exists 
  an epsilon-net with respect to cones of size $s(\epsilon)$ for every point
  set  of lower envelope points $P'\subseteq\setR^3$ in non-$\C$-degenerate position.
\end{lemma}
\begin{proof}
  The proof follows easily from the fact that the image of a cone $C$
  under the projection $\tau$ contains exactly those points that are the
  image of the points that are contained in $C$. 
\end{proof}

We refer to a cone $C$ as the corresponding cone of the object
$D=\tau(C)$. We will prove a few useful properties of the so constructed
set system $(S, \D)$. 

Notice, that the intersection of two triangular cones is again a
cone. Furthermore, the intersection of a possibly infinite family of
triangular cones is either empty or again a triangular cone since all
cones are closed. The intersection of the boundary of a cone with the
flattened lower envelope is either empty or a set of line segments
that form one simple closed cycle. Hence, the image of a cone under
the projection $\tau$ is a closed and connected region whose boundary is
a closed and connected cycle. 
  
\vspace{-3ex}
\begin{wrapfigure}[11]{r}{4.3cm}
  \begin{center}
    \epsfig{file=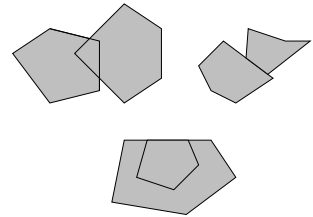,width=4cm}
    \caption{A set of non-piercing objects}
    \label{fig:piercing1}
  \end{center}
\end{wrapfigure}
.
\begin{definition}
  Two geometric objects(sets) $A\subseteq \setR^2$ and $B\subseteq \setR^2$ that are bounded
  by Jordan curves are said to be 
  \emph{non-piercing} 
  if 
  the boundary of $A$ and $B$ cross at most twice.
  A family of geometric objects is called non-piercing if every
  two objects from this family are non-piercing. See
  figure~\ref{fig:piercing1} for an illustration. 
\end{definition}


\begin{lemma}
  \label{lem:piercing}
  The projection $\tau$ produces a family $\D$ of non-piercing objects.
\end{lemma}

\begin{proof}
  Consider two cones $C_1$ and $C_2$ that intersect each other. If one
  is contained in the other, i.e. $C_1\subseteq C_2$ then we are done, as
  $\tau(C_1)\subseteq \tau(C_2)$ and hence their boundaries cannot cross. So if
  $C_1$ and $C_2$ intersect and none is subset of the other then the
  intersection of their boundaries are two rays emitting from the same
  point. Each of these rays intersects the flattened lower envelope
  exactly once. Hence, as the projection $\tau$ is injective the boundary
  of the two images of the cones $C_1$ and $C_2$ under the projection
  $\tau$ intersect exactly twice. Thus, the objects are non-piercing.   
\end{proof}

\subsection{Small Epsilon-Nets for Non-Piercing Objects in $\setR^2$}

In this subsection we will derive a few properties of the projection
that are necessary to apply the algorithm of Matou\v{s}ek et
al.~\cite{MSW90} for finding a small epsilon-net for
pseudo-disks. These properties also hold in general for any family of
non-piercing objects with the additional property that for any three
points there always exists an object that has these three points on
its boundary. 
However the proofs are a bit more involved. 
Since this does not lie in the scope of
this paper, we omit this here and focus only on the special family of
non-piercing objects that is produced by the projection described
above.

Consider the family of all cones that have $p$ and $q$ on its
boundary. The intersection of all these cones is a cone $C_{pq}$ that
has $p$ and $q$ on its boundary. Connect $p$ and $q$ by a Jordan curve
$E_{pq}$ such that it lies entirely in the cone $C_{pq}$ and on the
flattened lower envelope, for instance part of the boundary of
$C_{pq}$ that intersects the flattened lower envelope. The image of
$E_{pq}$ under the projection $\tau$ is a curve $\tau(E_{pq})$ embedded in
the plane.

\begin{defi}
  Let $\D$ be a family of non-piercing objects and let $S\subseteq \setR^2$ be a
  finite set of points.  
  We call two points $p, q\in \setR^2$ \emph{$\D$-Delaunay
    neighbors} if there exists an object $D\in \D$ that has $p$ and $q$ on
  its boundary and no other point of $S$ in its interior. 
  The
  $\D$-Delaunay graph of $S$, in short $\D$-$\DT(S)$, is the graph that
  is embedded in the plane, has $S$ as its vertex set and the edges
  $\tau(E_{pq})$  
  between all $\D$-Delaunay neighbors $p$ and $q$. 
\end{defi}
Due to the definition of the $\D$-Delaunay edge between two
$\D$-Delaunay neighbors $p$ and $q$ it is guaranteed that whenever a
object $D\in \D$ contains $p$ as well as $q$ then it also must contain
the $\D$-Delaunay edge $\tau(E_{pq})$.  
In the following we will prove
that this $\D$-Delaunay graph is in fact a triangulation of the vertex
set $S$. 

\begin{lemma}
  \label{lem:tria}
  The $\D$-Delaunay graph of the given finite point set $S$ in
  non-$\D$-degenerate position is a triangulation.  
\end{lemma}
\begin{wrapfigure}[13]{r}{5.2cm}
    \vspace{-4ex}
  \begin{center}
    \psfrag{p}{$p$}
    \psfrag{q}{$q$}
    \psfrag{r}{$r$}
    \psfrag{s}{$s$}
    \epsfig{file=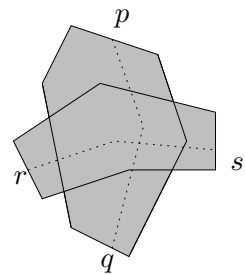,width=3cm}
    \caption{Two intersecting $\D$-Delaunay edges and their defining
      objects} 
    \label{fig:piercing2}
  \end{center}
\end{wrapfigure}
.\vspace{-4ex}
\begin{proof}
  First, we will prove that $\D$-$\DT(S)$ is planar. Suppose
  otherwise, i.e. two edges $\tau(E_{pq})$ and $\tau(E_{rs})$ intersect each
  other in the plane. Since the cone $C_{pq}$ does not have any point
  in its interior and $C_{rs}$ also does not have any point in its
  interior and since each of these cones has at most $3$ points on its
  boundary the objects $\tau(C_{pq})$ and $\tau(C_{rs})$ would have to
  pierce each other, see figure~\ref{fig:piercing2} for an
  illustration. Here, it is actually essential, that the set $S$ is in
  non-$\D$-degenerate position. Thus, the graph is planar. 

  The graph $\D$-$\DT(S)$ itself consists of an outer face which is
  defined by cones of the lower envelope that have at most 2 points on
  their boundary and all other faces are triangles defined by the
  cones of the lower envelope that have exactly three points on its
  boundary.  
  Suppose an inner face $F$ is not bounded by a triangle. Then, one
  can place the apex of a cone in such a way onto the flattened lower
  envelope such that its image under the projection $\tau$ is a point
  which lies inside this face $F$. By moving the cone upward one can
  ensure that the cone will finally have three points on its boundary
  whose image under the projection $\tau$ are three vertices of the face
  $F$ but no point in its interior. Hence, the face $F$ must be
  bounded by a triangle.   
  Hence, $\D$-$\DT(S)$ is a triangulation of the set $S$. 
\end{proof}

We call the points of $S$ that lie define the outer face the
\emph{convex hull of $S$ with respect to cone $C$} and we denote it by
$\con_C(S)$. 
It is a generalization of the standard convex hull and we will make
use of it later. 
For a standard triangulation one requires that the outer face is
determined by the convex hull. Here, we replaced the standard convex
hull by the convex hull with respect to cone $C$. This is the
appropriate generalization that we need.

\begin{lemma}
  \label{lem:connect}
  Let $D$ be an object produced by the projection $\tau$. The subgraph
  $G$ of $\D$-$\DT(S)$ induced by the points of $S$ that lie in $D$ is
  connected. 
\end{lemma}

\begin{proof}
  We prove the connectivity using induction over the number of points
  that lie in $D$. 
  If $D$ contains at most 2 points that it must be connected by
  definition and the fact that we can slide down the corresponding
  cone until both points lie on the boundary. So lets assume that
  every object $D$ that contains at most $k$ points from the set $S$
  induces a connected subgraph $G$. Now consider an object $D$ that
  contains $k+1$ points of $S$. Consider the cone that is the
  intersection of all cones that contain  exactly those $k+1$
  points. This cone has exactly three points on its boundary. We can
  move the cone by a small $\delta$ in such a way that each of the three
  points can be excluded separately. As all of these induced graphs
  are connected by induction hypothesis, the whole subgraph induced by
  $D$ must be connected.    
\end{proof}

We need two more lemmas. Both lemmas basically rely on the fact that
projection $\tau$ is continuous. 
\begin{lemma}
  \label{lem:cont}
  Let $S$ be a finite point set.
  \begin{enumerate}
  \item For any object $D\in \D$, there exists an object $D'\in \D$ with
    $S\cap D' = S\cap \intr D' = S\cap D$.
  \item For any object $D'\in \D$, there exists an object $D\in \D$ with
    $S\cap D' = S\cap \intr D' = S\cap \intr D$.   
  \end{enumerate}
\end{lemma}

\begin{proof}
  Let $C$ be the corresponding cone of $D$. If we move $C$ upward
  along the internal ray $r$ by a small $\delta$ then the corresponding
  object $D'$ of this cone will satisfy (1). On the other hand, if we
  move the cone $C$ downward along the ray $r$ by a small $\delta$ then the
  corresponding object $D'$ will satisfy (2).  
\end{proof}

\begin{lemma}
  \label{lem:dedge}
  Let $S$ be a finite point set in non-$\D$-degenerate position, let
  $(p, q)$ be a $\D$-Delaunay edge in $\D$-$\DT(S)$. Then, there exists
  an object $D$ with $p$ and $q$ on its boundary and with $S\cap D
  =\{p,q\}$.  
\end{lemma}

\begin{proof}
  Let $D$ be the object that assures that ${p, q}$ is a $\D$-Delaunay
  edge, i.e. $D$ has $p$ and $q$ on its boundary. Since the point set
  $S$ is in non-$\D$-degenerate position $D$ has at most three points
  on its boundary. If $D$ has exactly two points on its boundary we
  are done. So lets assume that $D$ has exactly three points on its
  boundary. Let $C$ be the corresponding cone of $D$ and let the
  corresponding points of $p$ and $q$ be $p'\in \setR^3$ and $q'\in \setR^3$.  
  Neither $p'$ nor $q'$ can lie on the intersection of two of the
  defining planes of cone $C$ because otherwise the cone could still
  be moved in an upward direction such that all three points still lie
  on the boundary until the cone hits a fourth point. But this would
  mean that the point set was in $\C$-degenerate position. Hence, $p'$
  and $q'$ lie in the interior of two of the plane segments of cone
  $C$.  
  If we now move the cone $C$ downward by a small $\delta$ such that it
  still touches $p'$ and $q'$ then the corresponding object of this
  cone will only have $p$ and $q$ on its boundary.  
\end{proof}

Having these properties, we can basically directly apply the algorithm
for finding an epsilon-net for pseudo-disks  
from~\cite{MSW90}. We will describe the algorithm here and prove its
correctness for our setting. 

We are given a finite point set $S$ in non-$\D$-degenerate position
and we want to find a subset $N\subseteq S$ of size $\bigO(1/ \epsilon)$ that stabs
any object $D$ that contains at least $\epsilon n$ points of $S$. 

Let $\delta=\epsilon/6$. First, let $S_1, \ldots, S_j$ be pairwise disjoint subsets of
$S$ with the following properties: Each $S_i$ contains $\delta n$ points,
their union contains the convex hull of $S$ with respect to cone $C$,
i.e. $\con_C(S)\subseteq \bigcup_{1\leq i\leq j} S_i$ and each $S_i$ is representable by
$S\cap \tau(C_i)$ for an appropriate  cone $C_i$. Such sets can be easily
constructed by repeatedly biting off points from $\con_C(S)$ with a
suitable cone $C_i$. Notice, that all these objects $D_i=\tau(C_i)$
belong to the collection $\D$.  

Next, find a maximal pairwise disjoint collection $S_{j+1}, \ldots , S_k$
of subsets of the remaining points $S\setminus \bigcup_{1\leq i\leq j} S_i$ satisfying
$S_i=S\cap D_i$ for some object $D_i$ and each subset containing $\delta n$
points. Obviously, there are at most $1/ \delta +1$ many subsets $S_i$ in
total. 
For an illustration we refer to figure~\ref{fig:delauny1}. 
We assign all points in $S_i$ the color $i$ and call all other points
\emph{colorless}. Let $\bar{S}$ be the set of all colored points. 
Note, that if an object contains only colorless points then it
contains less that $\delta n$ points, since the collection of subsets $S_i$
was maximal. 
\begin{figure}
\begin{minipage}{0.49\textwidth}
  \begin{center}
    \epsfig{file=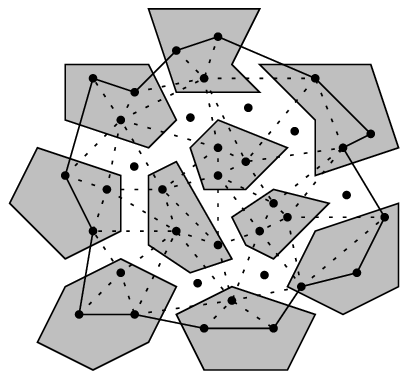,width=4cm}
    \caption{The sets $S_i$ and the convex hull $\con_C(S)$ with
      respect to cone $C$. The $\D$-Delaunay triangulation is drawn
      dotted.} 
    \label{fig:delauny1}
  \end{center}
\end{minipage}%
\hfill
\begin{minipage}{0.49\textwidth}
  \begin{center}
    \psfrag{R}{$R$}
    \epsfig{file=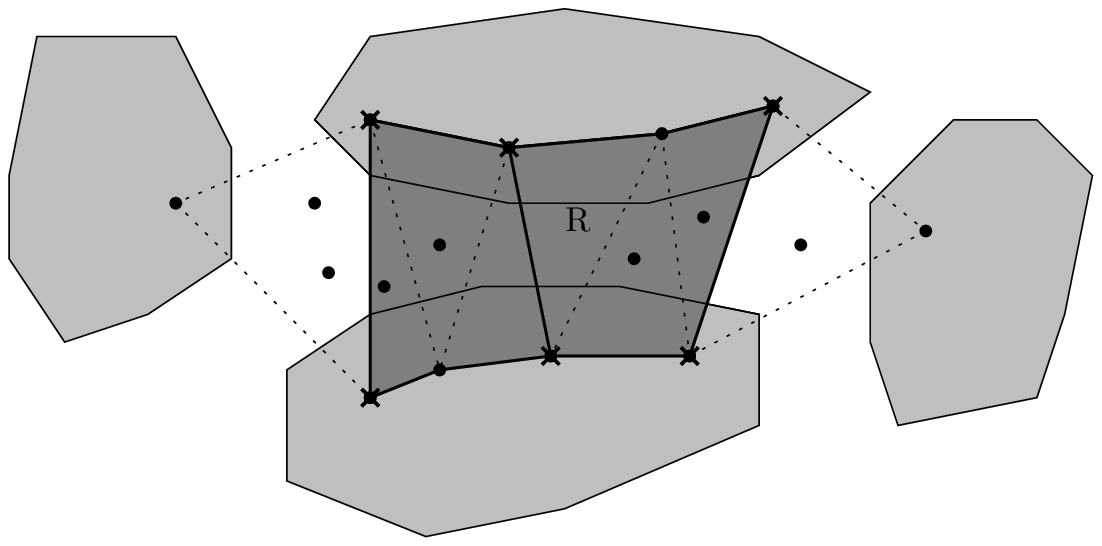,width=5cm}
    \caption{The corridor $R$ which is split into two sub-corridors
      and two tri-colored triangles. The corners of the sub-corridors
      are marked by crosses.} 
    \label{fig:delauny2}
  \end{center}
\end{minipage}
\end{figure}

Let $G$ be the $\D$-Delaunay graph of the set of colored points
$\bar{S}$, i.e. $G=\DT(\bar{S})$. $G$ is indeed a triangulation
(cf. lemma~\ref{lem:tria}). In this graph we call a triangle
\emph{uni-colored}, \emph{bi-colored} or \emph{tri-colored} depending
upon the number of colors its vertices have. In a similar way we call
edges uni-colored or bi-colored. We call a maximal connected chain of
bi-colored triangles in $G$ sharing bi-colored edges a \emph{corridor}
(cf. figure~\ref{fig:delauny2}).  
Since the graph $G$ is planar and each of the induced subgraphs $G\cap
D_i$ is connected according to lemma~\ref{lem:connect} the number of
such corridors is at most $3k-6$ (\cite{MSW90}). All colorless points
are contained in the corridors and the tri-colored triangles because
any uni-colored triangle is contained it its color-defining object. We
break each corridor $R$ into a minimum number of \emph{sub-corridors},
i.e. sub-chains of the chain that forms $R$, so that each sub-corridor
contains at most $\delta n$ colorless points. Since there are less than $n$
colorless points and since the total number of corridors is $3k-6$ the
total number of sub-corridors is $\bigO(1/ \delta)$.  

Each sub-corridor is bounded by two chains of uni-colored edges which
we call \emph{sides} and by two bi-colored edges which we call
\emph{ends} of the sub-corridor. The endpoints of the sides are called
\emph{corners}. Let $N\subseteq S$ be the set of all corners of all
sub-corridors. Since each sub-corridor has at most 4 corners the size
of $N$ is $\bigO(1/ \epsilon)$. 
The set $N$ is an epsilon-net for the set of non-piercing objects
$\D$.

The proof that $N$ is indeed an epsilon-net relies in principle on the
fact that the collection $\D$ are non-piercing objects and follows
along the lines of~\cite{MSW90}. 
\begin{proof}
  Let $D$ be an object that has no points of $S$ on its boundary
  (cf. lemma~\ref{lem:cont}) and assume that $D$ does not contain any
  points from $N$. The theorem is proven when we can show that $D$
  then contains less than $\epsilon n$ points of $S$. If $D$ contains no
  colored point then we are done, because the sets $S_i$ were a
  maximal. Hence, $D$ must contain at least one colored point. If it
  contains two colored points, lets say $z_1$ of color 1 and $z_2$ of
  color 2, we can draw the following picture: Let $D_1$ be the color
  defining object of color 1 and $D_2$ the color defining object of
  color $2$. Then $D$ intersects $D_1$ and $D_2$ but cannot pierce
  them. The area between $D_1$ and $D_2$ is a sub-corridor whose ends
  we denote by $(a_1, a_2)$ and $(b_1, b_2)$. Lemma~\ref{lem:dedge}
  assures that there is an object $D_a$ that has $a_1$ and $a_2$ on
  its boundary and there is an object $D_b$ that has $b_1$ and $b_2$
  on its boundary. Since $D$ also does not contain any point from $N$
  which are the corners of the sub-corridors, i.e. it does not contain
  $a_1, a_2, b_1$ or $b_2$ and since $D$ and $D_a$ as well as $D$ and
  $D_b$ are non-piercing it must lie between two ends of one
  sub-corridor. See figure~\ref{fig:delauny3} for an
  illustration. Now, as all objects $D_1$, $D_2$, $D_a$ and $D_b$
  contain at most $\delta n$ points and the sub-corridor also contains at
  most $\delta n$ points $D$ can contain at most $5\cdot \delta n = 5/6 \epsilon n < \epsilon n$
  points of $S$. 
\begin{figure}
\vspace{-2ex}
\begin{minipage}{0.45\textwidth}
  \begin{center}
    \psfrag{D}{$D$}
    \psfrag{D1}{$D_1$}
    \psfrag{D2}{$D_2$}
    \psfrag{Da}{$D_a$}
    \psfrag{Db}{$D_b$}
    \psfrag{a1}{$a_1$}
    \psfrag{a2}{$a_2$}
    \psfrag{b1}{$b_1$}
    \psfrag{b2}{$b_2$}
    \epsfig{file=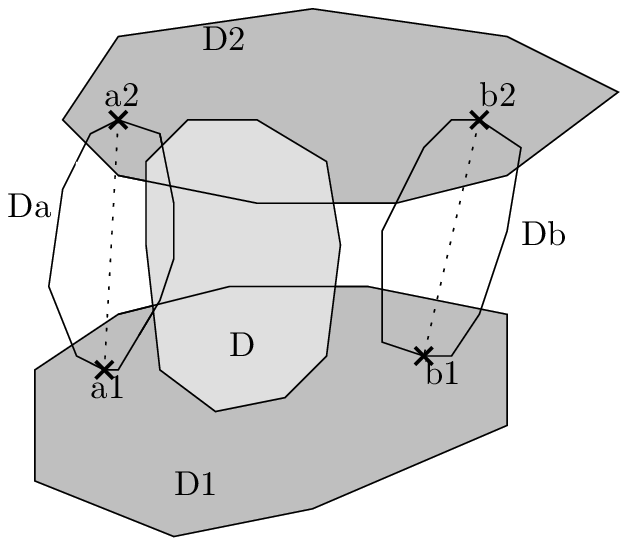,height=3cm}
    \vspace{-1ex}
    \caption{The case where $D$ contains colored points of at least
      two colors.} 
    \label{fig:delauny3}
  \end{center}
\end{minipage}%
\hfill
\begin{minipage}{0.45\textwidth}
  \begin{center}
    \psfrag{D}{$D$}
    \psfrag{D1}{$D_1$}
    \psfrag{D2}{$D_2$}
    \psfrag{Da}{$D_a$}
    \psfrag{Db}{$D_b$}
    \psfrag{a1}{$a_1$}
    \psfrag{a2}{$a_2$}
    \psfrag{b1}{$b_1$}
    \psfrag{b2}{$b_2$}
    \epsfig{file=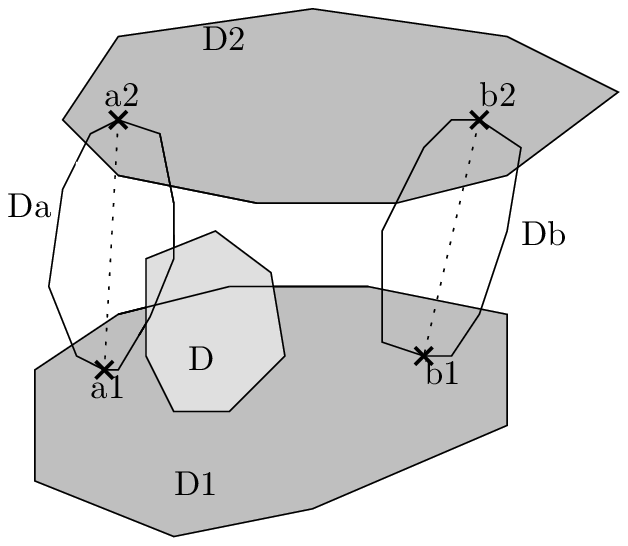,height=3cm}
     \vspace{-1ex}
    \caption{The case where $D$ contains colored points of exactly one
      color.} 
    \label{fig:delauny4}
  \end{center}
\end{minipage}
\end{figure}

The case where $D$ only contains points of one color and colorless
points is very similar. There is basically only one setup and it is
depicted in figure~\ref{fig:delauny4}. Arguing as above it easy to see
in this case that $D$ cannot contain more than $4\cdot \delta n < \epsilon n$ points
from $S$.  
\end{proof}
Hence, we have the following theorem
\begin{theorem}
  Let $\D$ be the set of non-piercing objects in $\setR^2$,
  that is produced by the projection $\tau$. For every finite point set
  in non-$\D$-degenerate position there exists an epsilon-net of size 
  $\bigO(1/ \epsilon)$. 
\end{theorem}
Together with lemma~\ref{lem:1} and lemma~\ref{lem:2} this
immediately implies our main theorem  
\begin{theorem}
  Given a finite point set $P\subseteq \setR^3$ and a polytope $T\subseteq \setR^3$. The set
  system $(P, \T)$ induced by a set of translates of polytope $T$ 
  admits an epsilon-net of size $\bigO(1/ \epsilon)$. 
\end{theorem}

\vskip-0.3cm
\section{From Epsilon-Nets to Hitting Sets}
In this section we will describe a constant factor approximation
algorithm to the hitting set problem using the epsilon-net of size
$\bigO(1/ \epsilon)$ from the previous section. 
Recall that in the hitting set problem we are given a set of points
$P\in \setR^3$ and a set polytopes  that are all translates of the same
polytope and we would like to select a subset $H\subseteq P$ of the input
points of minimal cardinality such that every polytope is stabbed by a
point in $H$. We denote the corresponding set system by $(P, \T)$. 
The fractional hitting set problem is a relaxation of the original
hitting set problem and is defined by the following linear program: 
\begin{eqnarray}
  \min & \sum_{p\in P} x(p) \\
  \st & \forall T \in \T & \sum_{p\in T} x(p) \geq 1 \\
  & \forall p\in P & x(p)\geq 0
\end{eqnarray} 

Let $\opt$ denote the optimal size of the hitting set and $\opt^*$ the
optimal value of the fractional hitting set problem. It is known that
the integrality gap is constant for set systems that admit an
epsilon-net of size $\bigO(1/ \epsilon)$~\cite{PA95}.  

Let $w: P \to \setR_{\geq 0}$ be a weight function for the set $P$. We define
the weight $w(A)$ of a subset $A\subseteq P$ to be the sum of the weights of
the elements of $A$. The weighted version of an epsilon-net is as
follows: 
\begin{defi}
  Consider a set system $(P, \T)$ and a weight function $w: P\to \setR_{\geq
    0}$. A set $H\subseteq P$ is called an \emph{epsilon-net with respect to 
    $w$} if $H\cap R\neq \emptyset$ for every subset $T\in \T$ for which $w(T)\geq \epsilon \cdot
  w(S)$. 
\end{defi}
There are algorithms that compute a hitting set provided one has an
epsilon-net finder.  
The core idea to all these algorithms is to find a weight function $w
: P\to \setR_{\geq 0}$ that assigns weights to the elements of $P$ and finds an
appropriate $\epsilon$ such that every set in $\T$ has weight at least $\epsilon \cdot
w(S)$. Once such weights are found it is then obvious that an
epsilon-net to this set system is automatically a hitting set. 

The algorithm given by Br\"{o}nnimann and Godrich~\cite{BG94} computes
these weights iteratively. Initially, all elements have weight
$1$. Then, in each iteration an epsilon-net is computed and then
checked whether it is also a proper hitting set. If not, i.e. there is
a set which is not hit, then the weights of its elements are
doubled. This is done until a hitting set is found. This algorithm can
be seen as a deterministic analogue of the randomized natural
selection technique used for instance by Clarkson~\cite{C95}. 

Another algorithm is by Even et al.~\cite{ERS05}. Here, the weights of
the elements and $\epsilon$ are directly found by the following linear
program: 
\begin{eqnarray}
  \max \epsilon \\
  \st & \forall T\in \T & w(T)\geq \epsilon  \\
  &  & \sum_{p\in P} w(p) = 1 \\
  & \forall p\in P & w(p)\geq 0 
\end{eqnarray} 

It suffices to approximate the solution to this linear problem. There
are numerous algorithms that find an approximate solution to such a
covering linear program efficiently~\cite{Y95, GK98}.

One can reduce the problem of finding a weighted epsilon-net to the
unweighted case. One just makes multiple copies of a point according
to its assigned weight and it can be shown that the cardinality of
this multiset can be bounded by $2n$~\cite{CV05}. 
Hence, an $\frac{\epsilon}{2}$-net for this set
system gives a hitting set for the original hitting set problem.
Hence, we have
\begin{theorem}
  There exists a polynomial time algorithm that computes a
  constant-factor approximation to the hitting set problem for
  translates of polytopes in $\setR^3$. 
\end{theorem}

\vskip-0.3cm
\section{From Hitting Set to Set Cover}

\begin{defi}
  The \emph{dual set system} of a set system $(P, \T)$ is the set
  system $(\T, P^*)$ where $P^*=\{\T_p : p\in P\}$ and $\T_p$ consists of
  all subsets of $\T$ that contain $p$. 
\end{defi}



Obviously, a set cover for the primal set system is a hitting set for
the dual set system. Hence, in order to solve the set cover problem
for a set system it suffices to solve the hitting set problem for the
dual set system. For arbitrary set systems, the dual set system can be
of quite different structure. In general it is only known that the
VC-dimension of the dual set system is less than $2^{d+1}$, where $d$
is the VC-dimension of the primal set system~\cite{A83}.  

However, we observe that if the set system is induced by translates of
a polytope, then the dual is again induced by  translates of a
polytope. 
To see this, let $(P, \T)$ be the primal set system. One just reduces
each polytope $T\in \T$ to a point, for instance each to its lowest
vertex. Let this be the set $P'$. Then, replace each point of $P$ by a
translate of the polytope $T'$ which is the inversion of $T$ in a
point. One easily verifies that the so constructed set system $(P',
\T')$ of points $P'$ and collection of translates of polytope $T'$ is
indeed equivalent to the dual $(\T, P^*)$.      
This holds in fact for all $\setR^d$.    
Hence, we can find a constant-factor approximation to the set cover
problem for translates of a polytope in $\setR^3$ in polynomial time. 
This brings us to our final theorem
\begin{theorem}
  There exists a polynomial time algorithm that computes a
  constant-factor approximation to the set cover problem for
  translates of polytopes in $\setR^3$. 
\end{theorem}

\vskip-0.3cm
\section{Conclusions and Open Problems}
In this paper we have given the first constant-factor approximation
algorithm for finding a set cover for a set of points in $\setR^3$ by a
given collection of  translates of a polytope as well as the first
constant-factor approximation algorithm for the corresponding hitting
set problem.  We achieved this result by providing an epsilon-net of
size $\bigO(\frac{1}{\epsilon})$ for the corresponding set system which is
optimal up to a multiplicative constant. 
Eventhough we can approximate a unit ball in $\setR^3$ up to any given
precision by a polytope, the corresponding question, whether there
exists a constant-factor approximation algorithm for unit balls in
$\setR^3$ still remains open.    

\vskip-0.3cm
\section*{Acknowledgements}
The author would like to thank Nabil H. Mustafa and Saurabh Ray for
useful discussion on the topic and an anonymous referee for pointing
out an error in a preliminary version of this paper.






\vskip-0.3cm
\begin{thebibliography}{10}

\bibitem{A83}
P.~Assouad.
\newblock {Densit\'e et dimension.}
\newblock {\em Ann. Inst. Fourier}, 33(3):233--282, 1983.

\bibitem{BEHW89}
A.~Blumer, A.~Ehrenfeucht, D.~Haussler, and M.~K. Warmuth.
\newblock Learnability and the vapnik-chervonenkis dimension.
\newblock {\em J. ACM}, 36(4):929--965, 1989.

\bibitem{BG94}
H.~Br\"{o}nnimann and M.~T. Goodrich.
\newblock Almost optimal set covers in finite vc-dimension: (preliminary
  version).
\newblock In {\em SoCG '94}, pages 293--302, New York, NY, USA, 1994. ACM
  Press.

\bibitem{C95}
K.~Clarkson.
\newblock Las vegas algorithms for linear and integer programming when the
  dimension is small.
\newblock {\em J. ACM}, 42(2):488--499, 1995.

\bibitem{CV05}
K.~L. Clarkson and K.~Varadarajan.
\newblock Improved approximation algorithms for geometric set cover.
\newblock In {\em SoCG '05}, pages 135--141, New York, NY, USA, 2005. ACM
  Press.

\bibitem{EW85}
H.~Edelsbrunner and E.~Welzl.
\newblock On the number of line separations of a finite set in the plane.
\newblock {\em J. Comb. Theory, Ser. A}, 38(1):15--29, 1985.

\bibitem{ERS05}
G.~Even, D.~Rawitz, and S.~Shahar.
\newblock Hitting sets when the vc-dimension is small.
\newblock {\em Inf. Process. Lett.}, 95(2):358--362, 2005.

\bibitem{GK98}
N.~Garg and J.~K\"{o}nemann.
\newblock Faster and simpler algorithms for multicommodity flow and other
  fractional packing problems.
\newblock In {\em FOCS '98}, page 300, Washington, DC, USA, 1998. IEEE Computer
  Society.

\bibitem{HM85}
D.~S. Hochbaum and W.~Maass.
\newblock Approximation schemes for covering and packing problems in image
  processing and vlsi.
\newblock {\em J. ACM}, 32(1):130--136, 1985.

\bibitem{KPW92}
J.~Koml\'{o}s, J.~Pach, and G.~J. Woeginger.
\newblock Almost tight bounds for epsilon-nets.
\newblock {\em Discrete and Computational Geometry}, 7:163--173, 1992.

\bibitem{LY94}
C.~Lund and M.~Yannakakis.
\newblock On the hardness of approximating minimization problems.
\newblock {\em J. ACM}, 41(5):960--981, 1994.

\bibitem{M}
J.~Matousek.
\newblock {\em Lectures on Discrete Geometry}.
\newblock Springer-Verlag New York, Inc., Secaucus, NJ, USA, 2002.

\bibitem{M92}
J.~Matou\v{s}ek.
\newblock Reporting points in halfspaces.
\newblock {\em Comput. Geom. Theory Appl.}, 2(3):169--186, 1992.

\bibitem{MSW90}
J.~Matou\v{s}ek, R.~Seidel, and E.~Welzl.
\newblock How to net a lot with little: Small epsilon-nets for disks and
  halfspaces.
\newblock In {\em SoCG '90}, pages 16--22, 1990.

\bibitem{PA95}
J.~Pach and P.~K. Agarwal.
\newblock {\em Combinatorial Geometry}.
\newblock Wiley, New York, 1995.

\bibitem{PW90}
J.~Pach and G.~Woeginger.
\newblock Some new bounds for epsilon-nets.
\newblock In {\em SoCG '90}, pages 10--15, New York, USA, 1990. ACM Press.

\bibitem{VC71}
V.~N. Vapnik and A.~Ya. Chervonenkis.
\newblock On the uniform convergence of relative frequencies of events to their
  probability.
\newblock {\em Theory Probab. Appl.}, 16:264--280, 1971.

\bibitem{Y95}
N.~E. Young.
\newblock Randomized rounding without solving the linear program.
\newblock In {\em SODA '95}, pages 170--178, Philadelphia, PA, USA, 1995.
  Society for Industrial and Applied Mathematics.

\end{thebibliography}
\bibliographystyle{plain}

\end{document}